\RequirePackage{amsmath}
\documentclass{llncs}

\usepackage[english]{babel}
\usepackage[utf8x]{inputenc}
\usepackage[T1]{fontenc}

\usepackage{graphicx}
\usepackage[colorinlistoftodos]{todonotes}
\usepackage[colorlinks=true, allcolors=blue]{hyperref}
\usepackage{authblk}

\title{On the Price of Anarchy of Cost-Sharing in Real-Time Scheduling Systems}

\author{Eirini Georgoulaki\inst{1} \and Kostas Kollias\inst{2}}
\institute{University of Athens, \email{eirini.geo.98@gmail.com}
\and 
Google Research, \email{kostaskollias@google.com}}

\date{}

\begin{document}
\maketitle

\begin{abstract}
We study cost-sharing games in real-time scheduling systems where the activation cost of the server at any given time is a function of its load. We focus on monomial cost functions and consider both the case when the degree is less than one (inducing positive externalities for the jobs) and when it is greater than one (inducing negative externalities for the jobs). For the former case, we provide tight price of anarchy bounds which show that the price of anarchy grows to infinity as a polynomial of the number of jobs in the game. For the latter, we observe that existing results provide constant and tight (asymptotically in the degree of the monomial) bounds on the price of anarchy. We then switch our attention to improving the price of anarchy by means of a simple coordination mechanism that has no knowledge of the instance. We show that our mechanism reduces the price of anarchy of games with $n$ jobs and unit activation costs from $\Theta(\sqrt{n})$ to $2$. We also show that for a restricted class of instances a similar improvement is achieved for monomial activation costs. This is not the case, however, for unrestricted instances of monomial costs for which we prove that the price of anarchy remains super-constant for our mechanism.
\end{abstract}

\section{Introduction}

The model of cost-sharing in real-time scheduling systems was introduced by \cite{T18} in response to the emergence and popularity of cloud computing as well as to the efforts to reduce power consumption in large computing systems \cite{A10,BGNS01,CGK14,IP05,KSST15}. In the model studied in \cite{T18} there is a server and a collection of $n$ jobs that are to be scheduled on the server. Each job $j$ has a release time $r_j$ and a deadline $d_j$. Time is slotted and each job gets to select which slot in $[r_j,d_j)$ it will be scheduled on. The server has a unit activation cost per time slot, which models the energy spent to keep the server open. This is an expressive model for many applications in cloud computing and data center optimization, which makes the understanding of the inefficiency of such a system and ultimately its improvement important tasks. In a more general model, the activation cost can depend on the load that the server has to process at any given time slot $t$, i.e., the energy spent will be a function $c(l_t)$ that depends on the number of jobs $l_t$ that are to be processed during $t$. This generalized model is the main focus of this work.

In the standard cost-sharing setting studied in \cite{T18}, each of the jobs processed at time $t$ assumes an equal share of the server activation cost $c(l_t)$ (or a proportional share for a more general setting where each job places a different load on the server). Given this rule, we would expect each job $j$ to optimize for its individual cost share and declare the slot that minimizes the ratio of the activation cost to the number of jobs (which is precisely the individual cost share) among all slots in its window $[r_j,d_j)$. When this is true for every job, we have an assignment that is a {\em Nash equilibrium} (NE). The inefficiency of a NE is captured by the {\em price of anarchy} (PoA) which is the worst case ratio of the total cost in a NE assignment over the total cost in the optimal assignment.

For the base case of unit activation costs we have $c(0)=0$ and $c(l)=1$ for every $l>0$. The PoA of such games was shown in \cite{T18} to be $\Theta(\sqrt{n})$ and the question of what happens when the server has an activation cost that depends on the load placed on it was posed as an open problem. A major part of our results focuses precisely on answering this question. We study cost functions of the form $c(x)=x^d$ for some parameter $d>0$. Note that for $d<1$ the cost share function $c(l)/l$ is monotonically decreasing and that for $d>1$ it is monotonically increasing. The game we study here belongs to the class of {\em congestion games} \cite{R73}, the PoA of which has received significant attention in the literature. Our games are equivalent to singleton congestion games with uniform resources arranged on a line and with strategy spaces that correspond to intervals on this line. We provide tight PoA bounds for the case of positive externalities (which is the best motivated case for our application of interest), i.e., $d<1$, showing that the PoA is $\Theta(n^{(1-d)/2})$. For $d>1$, the upper bound for general congestions games shows the PoA does not depend on the number of jobs and is at most $d^{\Theta(d)}$ (i.e., constant for fixed $d$). We observe that an early $d^{\Theta(d)}$ lower bound instance \cite{AAE05} designed for routing games applies to our setting, showing that the PoA is in fact $d^{\Theta(d)}$.

Subsequently, we focus on the design of {\em coordination mechanisms} for the problem in an effort to improve the PoA. A coordination mechanism \cite{CKN04} is a set of a-priori rules the system designer can set without knowledge of the instance. A coordination mechanism can, for example, modify the cost shares of jobs, expand or restrict their strategy spaces, etc. The modified rules of the game under the coordination mechanism will change the set of NE outcomes, hopefully to the improvement of the PoA. We design a simple coordination mechanism that merely expands the strategy space of each job by asking them to declare, not only a slot, but also a payment. It is clear that this mechanism requires no knowledge of the specifics of the instance and is very simple to implement. We show that this simple modification has a surprisingly strong impact on the PoA for the case of unit activation costs, reducing it from infinity --- specifically $\Theta(\sqrt{n})$ --- down to $2$. For monomial cost functions of degree less than $1$, we prove that our mechanism has super-constant PoA for general instances, however, we also prove that it does achieve an improvement from super-constant to constant for a special family of instances. Specifically, for instances such that the optimal solution uses a single slot (which includes, for example, instances with a common release time or a common deadline), our mechanism reduces the PoA from $\Theta(n^{(1-d)/2})$ to $\Theta(1)$. (Recall that for monomials of degree larger than $1$, it is known that the PoA of congestion games is already constant \cite{ADGMS11} without applying any coordination mechanism.)
\subsection{Related Work}

The work most closely related to our paper is the one in \cite{T18} where the model we study was introduced and various results were obtained for the case of constant (but possibly time-dependent) server activation costs with respect to the price of anarchy and the related concepts of the strong price of anarchy (inefficiency of equilibria with respect to coordinated deviations) and the price of stability (inefficiency of the best Nash equilibrium as opposed to the worst). Our work also has ties to literature on the price of anarchy of congestion games, cost-sharing in congestion games, and coordination mechanisms, all three of which we discuss below.

Congestion games where introduced by Rosenthal in \cite{R73} as a class of games that guarantee the existence of (pure) Nash equilibria. In a congestion game, there is a set of resources that the participants can use. Each one of the participants has a strategy space that allows them to select one of given subsets of resources, something that allows for various expressive models, including routing in networks. The cost of each resource, which is a function of the number of participants using it, is distributed equally among its users. The price of anarchy of a special case of congestion games was first studied in \cite{KP09}, where the price of anarchy was first introduced. A long sequence of follow-up work gave a strong understanding of the price of anarchy in congestion games \cite{ADGMS11,ADKTWR08,AAE05,CK05,R15}. Generalizations to weighted models have also been studied \cite{ADGMS11,BGR14} albeit with the drawback that, in contrast to standard congestion games, existence of a (pure) Nash equilibrium is not guaranteed \cite{GMV05}.

Cost-sharing aspects in the study of congestion games were brought into the picture to correct for the absence of equilibria in weighted games. The work in \cite{KR15} shows how the Shapley value cost-sharing method can be applied to restore pure equilibria, whereas \cite{GMW14} shows that the more general class of generalized weighted Shapley values are the only ones that can guarantee this property. The price of anarchy of the induced games has been the subject of extensive study, with examples including \cite{GKK15,GKR16,KS15,MP17,MW13,RS16,FH13}.

A more general approach seeking to improve the price of anarchy is that of coordination mechanisms. A coordination mechanism gives the designer more freedom in designing the game. This freedom includes modifying the strategy spaces of the participants or changing how costs are defined on the resources. For example, coordination mechanisms were applied in a multiple machine makespan minimization scheduling setting in \cite{CKN04} by means of introducing different scheduling policies on different machines. Follow up work on coordination mechanisms includes \cite{AFJMS15,C13,CGMO15,ILMS09}.

\subsection{Summary of Our Results and Roadmap}

Section \ref{model} presents our model and notation. In Section \ref{decreasing} we focus on the case of monomial cost functions of degree $d<1$ and prove that the PoA is approximately $2 n^{(1-d)/2}/(1+d)$, where $n$ is the number of jobs in the game. In Section \ref{increasing} we discuss the case when $d>1$ and observe that early results on congestion games \cite{AAE05} can be used to show that the PoA is $d^{\Theta(d)}$. Finally, in Section \ref{mechanism}, we define our coordination mechanism and, in Section \ref{coordination-unit}, we prove that it achieves a strong improvement on the PoA for the case of unit activation costs (from $\Theta(\sqrt{n})$ to $2$). Switching to monomials with $d<1$, even though it is the case that the PoA grows to infinity for our mechanism as well, we are still able to show in Section \ref{coordination-monomial} that we can reduce the PoA from $\Theta(n^{(1-d)/2})$ to $\Theta(1)$ for the class of {\em common slot instances}, in which the optimal uses a single slot (note that this class includes, among others, instances with a common release time or a common deadline). As stated earlier, the PoA for the case with $d>1$ is already constant. Section \ref{conclusion} concludes the paper.
\section{Preliminaries}\label{model}

The game is specified by the following parameters: (a) a cost function $c(x)$ that gives the cost at any given time step as a function of the number of jobs $x$ that have to be processed at that time step, (b) a time horizon $T$ that specifies the available time slots as $t=1,2,\dots,T$, and (c) a set of jobs $J$, with each $j\in J$ having an integer release time $r_j$ and an integer deadline $d_j$ such that $0<r_j<d_j<T$. 

Let $s_j$ denote the slot declared by job $j$ such that $r_j\le s_j<d_j$ and let $s$ denote the vector of declared slots. The load on slot $t$, $l_t(s) = |\{j : s_j = t\}|$, is the number of jobs declaring slot $t$. The cost on slot $t$ is then $c(l_t(s))$ and the total cost is:
\[ C(s) = \sum_{t=1}^{T} c(l_t(s)). \]

An assignment $s$ is a {\em Nash equilibrium} (NE) when for every job $j$ we get:
\[ \frac{c(l_{s_j}(s))}{l_{s_j}(s)} \le \frac{c(l_t(s)+1)}{l_t(s)+1}, \]
for every $t\neq s_j$ in $[r_j,d_j)$. This expression suggests the cost share at the slot declared by $j$ is at most the cost share $j$ would get by deviating to any other slot in its interval. The inefficiency of equilibrium solutions is given by the {\em price of anarchy} (PoA), which is the worst case ratio of the total cost in a NE over the total cost in the optimal assignment:
\[ PoA = \frac{\max_{s\textrm{~a~NE}}C(s)}{\min_{s^*}C(s^*)}. \]
\section{Monomial Cost Functions}

\subsection{Decreasing Cost Shares: $d<1$}\label{decreasing}

In this section we analyze the PoA for the case when the cost function equals $x^d$ with $d<1$, for which the cost share $c(x)/x$ is strictly decreasing. The main result of the section is the following theorem, which we prove in two subsequent lemmas.

\begin{theorem}
For games with $n$ jobs and cost function $c(x)=x^d$ with $d<1$, the worst case PoA is $\Theta(n^{(1-d)/2})$.
\end{theorem}

\begin{lemma}\label{decreasing-upper-bound}
Given a game with $n$ jobs and cost function $c(x)=x^d$ with $d<1$, for any NE $s$ and any optimal assignment $s^*$ we get:
\[ \frac{C(s)}{C(s^*)} = O\left(n^{\frac{1-d}{2}}\right). \]
\end{lemma}
\begin{proof}
Consider an optimal assignment $s^*$ and focus on a slot $t$ that holds $l_t(s^*)=l_t$ jobs in $s^*$. We will write $J(t)$ for the set including these $l_t$ jobs. The jobs in $J(t)$ have a cost of $c(l_t)$ in $s^*$ and we wish to bound the cost they can have in any given NE $s$.

We begin by proving the claim that, in the NE $s$ and for any positive integer $x$, there exist at most $2$ slots that have $x$ jobs and include jobs from $J(t)$. We prove this claim by contradiction. Suppose there are at least $3$ slots with $x$ jobs that host jobs from $J(t)$. Suppose the median is $t'$. If $t'\ge t$, then we have at least two slots that are greater than or equal to $t$ and hold $x$ jobs including some from $J(t)$, otherwise we have at least two slots that are less than or equal to $t$ with this property. We will treat the case when there are two slots that are greater than or equal to $t$. The other case is symmetric. Call these slots $t'$ and $t''$ with $t''>t'$. Consider some job $j\in J(t)$ that uses $t''$ in $s$. The allowed interval of $j$ includes $t$, since it is scheduled on it in $s^*$, and $t''$, since it is scheduled on it in $s$. Since $t\le t' < t''$, it follows that $t'$ is also in the allowed interval for $j$. However, we can observe that $j$ has an incentive to deviate from $t''$ to $t'$ and improve its cost from $c(x)/x$ to $c(x+1)/(x+1)$. This contradicts the fact that $s$ is a NE and proves our original claim that at most $2$ slots can have jobs from $J(t)$ and exactly $x$ jobs.

Given the above, we return to the task of upper bounding the total payments of jobs in $J(t)$ in $s$. By the claim in the previous paragraph, at most $2x$ jobs from $J(t)$ can pay $c(x)/x$ in $s$. This means at most $2$ jobs can pay the maximum $c(1)/1$, at most $4$ jobs the second highest $c(2)/2$, etc. It follows that the total payments in $s$ of the jobs in $J(t)$ are upper bounded by:
\[ \sum_{j=1}^{h_t} 2c(j), \]
where $h_t$ is the smallest number such that $h_t^2+h_t\ge l_t$. Then the ratio of the total cost paid by the jobs in $J(t)$ in $s$ over the same cost in $s^*$ is at most:
\[ \frac{\sum_{j=1}^{h_t} 2c(j)}{c(l_t)}. \]
Suppose $t$ is in fact the slot that maximizes this ratio, meaning the above expression is an upper bound on the PoA. We get:
\begin{align}
\frac{C(s)}{C(s^*)} &\le \frac{\sum_{j=1}^{h_t} 2c(j)}{c(l_t)} = \frac{\sum_{j=1}^{h_t} 2 j^d}{l_t^d} \nonumber\\
&\le \frac{\int_{0}^{h_t+1}x^d dx}{l_t^d} \le \frac{2(h_t+1)^{1+d}}{(1+d)l_t^d} \nonumber\\
&\le \frac{2(h_t+1)^{1+d}}{(1+d)h_t^d(h_t-1)^d}\nonumber\\
&=O\left(h_t^{1-d}\right)=O\left(l_t^{\frac{1-d}{2}}\right)=O\left(n^{\frac{1-d}{2}}\right). \nonumber
\end{align}
This completes the proof of the upper bound.\qed
\end{proof}

\begin{lemma}\label{decreasing-lower-bound}
For every positive integer $h$ and for every cost function $c(x)=x^d$ with $d<1$, there exists a game with $n=h^2+h$ jobs, and a NE $s$ of that game, such that:
\[ \frac{C(s)}{C(s^*)} = \Omega\left(h^{1-d}\right)=\Omega\left(n^{(1-d)/2}\right), \]
where $s^*$ is an optimal assignment of that game.
\end{lemma}
\begin{proof}
Our instance has $n=h^2+h$ jobs and $2h+1$ slots. For ease of exposition we will shift time and call the slots $-h,-h+1,\ldots,-1,0,1,\ldots,h-1,h$. There exist $j$ jobs that can use slots $[-j,0]$ and $j$ jobs that can use slots $[0, j]$ for $j=1,2,\ldots,h$. Observe that slot $0$ is the only one that is common to all intervals. The optimal solution $s^*$ would place all jobs on $0$ for a total cost:
\[ C(s^*)=(h^2+h)^d=\Theta(n^{d}). \]

Consider the following assignment $s$, which as we will argue is a NE. Every one of the $j$ jobs with interval $[-j,0]$ selects slot $-j$ and every one of the $j$ jobs with interval $[0,j]$ selects slot $j$. Note that the number of jobs on slot $t$ is $|t|$. This fact implies the assignment is actually a NE, since the jobs on any slot $t$ share the slot with $|t|-1$ other jobs, while every other slot $t'$ between $t$ and $0$ has $|t'|\le |t|-1$ jobs. The cost of this assignment is:
\[ C(s)=\sum_{j=1}^{h}2j^d\ge 2\int_{1}^{h}x^d dx=\frac{2}{d+1}\left(h^{d+1}-1\right)=\Theta\left(n^{(d+1)/2}\right). \]
Taking the ratio $C(s)/C(s^*)$ completes the proof.\qed
\end{proof}

\begin{remark}
From the proofs of Lemmas \ref{decreasing-upper-bound} and \ref{decreasing-lower-bound} it follows that the worst case PoA is in fact approximately:
\[ \frac{2n^{\frac{1-d}{2}}}{1+d}. \]
\end{remark}
\subsection{Increasing Cost Shares: $d>1$}\label{increasing}

We now discuss the PoA for the case when the cost function equals $x^d$ with $d>1$. In this case the cost share $c(x)/x$ is strictly increasing. We observe that the general upper bound on the PoA of congestion games and an early lower bound designed for routing games which applies to our model yield the following theorem.

\begin{theorem}
(\cite{AAE05} Theorem 4.3) For games with cost function $c(x)=x^d$ with $d>1$, the worst case PoA is $d^{\Theta(d)}$.
\end{theorem}

Note that the PoA is constant when $d$ is fixed and is independent of the number of players $n$, in contrast to the unit and $d<1$ cases. In this work, we are mostly interested in this behavior of the PoA as a function of $n$. However, observing the PoA as a function of $d$ is also of interest and we note that the $d^{\Theta(d)}$ expression hides a gap. For example, for $d=2$, the lower bound of \cite{AAE05} is $2$, whereas the general upper bound for congestion games is $5/2$. The lower bound in \cite{AAE05} is a very natural instance, where, in the optimal solution, each slot has unit occupancy, whereas in the NE, slots become progressively more congested. This natural flavor of the instance could tempt us to conjecture that the lower bound is in fact tight for our setting, however, in the next result we show that this is not the case.

\begin{lemma}
For games with cost function $c(x)=x^2$, the worst case PoA is strictly larger than $2$.
\end{lemma}
\begin{proof}
The game has the following slots in order: first a slot with label ``$6$'', then $2$ slots with label ``$5$'', then $5$ slots with label ``$4$'', then $20$ slots with label ``$3$'', then $60$ slots with label ``$2$'', then $120$ slots with label ``$1$'', and, finally, another $120$ slots with label ``$0$''. The labels of slots signify precisely how many jobs are on them in the NE $s$. The allowable slots of a job on a slot with label ``$x$'' are all slots with labels from ``$6$'' up to ``$x-1$''. Clearly $s$ is a NE. Simple calculations show that the total cost in $s$ is $706$.

Now consider the outcome $s^*$ where all the jobs from slots with label ``$x$'', are spread as evenly as possible on the slots with label ``$x-1$''. Again simple calculations show that the total cost of $s^*$ is $352$, which gives a PoA bound $C(s)/C(s^*)>2$.\qed
\end{proof}

Identifying the precise PoA value as a function of $d$ appears to be a challenging open problem that we leave as future work.
\section{Coordination Mechanism}\label{mechanism}

In this section we design a coordination mechanism that applies a simple modification to the game, without knowing anything about the instance, and significantly improves the PoA. For the case of constant (unit) costs, we show that the PoA is brought down to 2. For $c(x)=x^d$ with $d<1$ we prove that the PoA is improved from super-constant to constant for the class of instances which have a single slot occupied in the optimal solution. This class includes instances with a common release time or a common deadline for all jobs, as well as instances with a batch of jobs centered around a common slot. We prove that, without this restriction, the PoA for $c(x)=x^d$ with $d<1$ remains super-constant even under our mechanism. Note that for $d>1$ the PoA is known to be constant even without applying any coordination mechanism.

Our mechanism changes the strategy space of each job $j$, from being simply a slot $s_j\in [r_j,d_j)$, to a pair $(s_j,\xi_j)$ where $s_j\in[r_j,d_j)$ is again the declared slot and $\xi_j\ge 0$ is a payment. The mechanism will open slot $t$ under strategies $(s,\xi)$ if and only if:
\[ \sum_{j:s_j=t}\xi_j\ge c(l_t(s)), \]
i.e., if the jobs selecting slot $s$ cover the activation cost with their declared payments. We assume every job $j$ has an infinite cost for not being processed (i.e., when slot $s_j$ is not opened), an assumption that is implicitly present in the base model as well, since jobs do not have the option of staying out of the game. In this framework, the NE condition asserts that each $(s_j,\xi_j)$ are such that:
\begin{align}\label{mechanism-nash-initial}
\sum_{j':s_{j'}=s_j}&\xi_{j'}\ge c(l_{s_j}(s))\nonumber\\
&\textrm{ and } \xi_j \le \max\left\{0, c(l_t(s)+1) - \sum_{j':s_{j'}=t}\xi_{j'}\right\}, \forall t\in[r_j,d_j)\setminus\{s_j\}\nonumber\\
&\textrm{ and } \xi_j \le \max\left\{0, c(l_{s_j}(s)) - \sum_{j':s_{j'}=s_j,j'\neq j}\xi_{j'}\right\}.
\end{align}
The first inequality guarantees slot $s_j$ is open, the second that there is no slot $t$ that job $j$ can move to, pay the minimum needed to open it (or keep it open), and get a lower payment, and the third that the job should pay as much as necessary to keep its current slot open.

\begin{lemma}\label{mechanism-lemma}
Every NE $(s,\xi)$ of our coordination mechanism is such that for every occupied slot $t$ we have $\sum_{j:s_j=t}\xi_j=c(l_t(s))$.
\end{lemma}
\begin{proof}
Suppose this is not the case. Then there exists some $t$ such that:
\[ \sum_{j:s_j=t}\xi_j>c(l_t(s)). \]
Consider any job $j$ such that $s_j=t$ and $\xi_j>0$. Clearly such a job must exist. We get:
\[ \sum_{j':s_{j'}=s_j}\xi_{j'}>c(l_{s_j}(s)) \Rightarrow \xi_j > c(l_{s_j}(s))-\sum_{j':s_{j'}=s_j,j'\neq j}\xi_{j'}, \]
which violates \eqref{mechanism-nash-initial} and gives a contradiction.\qed
\end{proof}

Using Lemma \ref{mechanism-lemma}, we may simplify the NE condition \eqref{mechanism-nash-initial} as:
\begin{equation}\label{mechanism-nash}
\sum_{j':s_{j'}=s_j}\xi_{j'}=c(l_{s_j}(s)) \textrm{ and } \xi_j \le c(l_t(s)+1) - \sum_{j':s_{j'}=t}\xi_{j'}, \forall t\in[r_j,d_j)\setminus \{s_j\}.
\end{equation}

\subsection{Unit Activation Costs}\label{coordination-unit}

\begin{lemma}\label{zero-payment-lemma}
Let $(s,\xi)$ be a NE of our coordination mechanism in a game with unit activation costs. Then for every job $j$, either $\xi_j=0$ or every slot in $[r_j,d_j)\setminus s_j$ is unoccupied.
\end{lemma}
\begin{proof}
Suppose $\xi_j>0$. Then, by \eqref{mechanism-nash}, we get that for any slot $t\in[r_j,d_j)\setminus s_j$:
\begin{equation}\label{lemma-condition}
1 - \sum_{j':s_{j'}=t}\xi_{j'} \ge \xi_j > 0.
\end{equation}
However, by Lemma \ref{mechanism-lemma}, we know that either $t$ is unoccupied or $\sum_{j':s_{j'}=t}\xi_{j'}=1$. From this, and given \eqref{lemma-condition}, we get that $t$ is unoccupied.
\qed
\end{proof}

\begin{theorem}
The PoA of our coordination mechanism for unit activation costs is $2$.
\end{theorem}
\begin{proof}
Focus on a slot $t$ and the set of jobs $J(t)$ that use it in a given optimal assignment. The total cost paid by these jobs in the optimal assignment is $1$. We will show that the same set of jobs pay at most $2$ in any NE $(s,\xi)$.

We examine two cases. First the case when slot $t$ is open in $(s,\xi)$. The jobs from $J(t)$ that use $t$ in $(s,\xi)$ pay at most $1$, as given by Lemma \ref{mechanism-lemma} (they might be paying less than $1$ as $t$ might have additional jobs outside $J(t)$). The jobs in $J(t)$ that use other slots pay $0$, as given by Lemma \ref{zero-payment-lemma} and the fact that all of these jobs have an occupied slot in their windows other than the one they are using, namely, slot $t$. This proves that, in this case, the total payments of jobs in $J(t)$ are at most $1$.

We now focus on the case when slot $t$ is not open in $(s,\xi)$. We first show that the jobs from $J(t)$ who are using slots larger than $t$ are paying a total of at most $1$. Focus on the smallest such slot $t'$. By the fact that $(s,\xi)$ is a NE and using Lemma \ref{mechanism-lemma} we get that the total payments on $t'$ are $1$. Now focus on any slot $t''>t'$ and any job $j\in J(t)$ that uses $t''$. Slot $t'$ must be in the window of job $j$ as both $t''$ and $t$ satisfy this property and $t'$ lies between them. Then, by Lemma \ref{zero-payment-lemma} we get that $\xi_j=0$. This proves that the total payment over all jobs $j$ that use slots larger than $t$ is $1$. The proof is symmetric for slots smaller than $t$, which shows that for every slot $t$ in an optimal solution and for the jobs $J(t)$ using it, the total payment of the jobs $J(t)$ in a NE is at most $2$. This directly implies an upper bound of $2$ for the PoA.

We now present a lower bound of $2$ for the PoA. There exist two jobs $j_1,j_2$. Job $j_1$ can be scheduled at times $1$ or $2$, whereas job $j_2$ can be scheduled at times $2$ or $3$. Assigning job $j_1$ to slot $1$, job $j_2$ to slot $2$, and setting $\xi_1=\xi_2=1$ gives rise to a NE with total cost $2$. It is easy to verify that this is a NE as each job has only one alternative slot, slot $2$, where again they would have to pay a unit cost. The optimal assignment places both jobs in slot $2$ for total cost $1$.\qed
\end{proof}

We close with a note on existence of a NE for our mechanism. In fact we prove that for any optimal solution there exist payments that will yield a NE.

\begin{theorem}
For unit activation costs, our coordination mechanism induces games such that, for every optimal assignment $s^*$, there exists a vector of payments $\xi$ such that $(s^*,\xi)$ is a NE. 
\end{theorem}
\begin{proof}
We first claim that an optimal assignment $s^*$ has the property that, on every occupied slot $t$, there exists at least one job such that no other slot in its allowed interval is occupied. If this is not true for some $t$, then moving all jobs $j$ with $s_j=t$ to some other occupied slot in their intervals will decrease the total cost by $1$, contradicting optimality of $s^*$. Given the above, we can find a job $j$ on every slot $t$ that has no other occupied slot in its interval and charge it the full cost of the slot by setting $\xi_j=1$. Every other job pays $0$. This solution will be a NE since the open slot costs are covered and all jobs who do not freeload, pay a unit cost and would pay the exact same cost if they were to deviate to any other slot in their intervals, since they are all unoccupied.\qed
\end{proof}

\subsection{Monomial Activation Costs with $d<1$}\label{coordination-monomial}

We first restrict ourselves to instances such that the optimal solution uses a single slot. Clearly this is the case if and only if there is some slot that is included in the interval of every job in the instance. This is the case, for example, when the jobs have common release times or deadlines. Note that by the instance in Lemma \ref{decreasing-lower-bound} (or simple modifications of it for the case of a common release time or a common deadline) we get that the PoA is infinite in the original game. In the next theorem we prove that our coordination mechanism reduces the PoA to a constant for every $d<1$ for the instances under consideration which we call {\em common slot instances}.

\begin{theorem}
For common slot instances and $c(x)=x^d$ with $d<1$, the PoA of our coordination mechanism is $\Theta(1)$.
\end{theorem}
\begin{proof}
For simplicity of exposition, suppose that, in a given game, we shift time so that the only slot used by the optimal solution $s^*$ is slot $0$. Let $s$ be a NE of the game. For simplicity and without loss of generality (due to symmetry), suppose the positive slots have a larger or equal cost compared to the non-positive slots in $s$. If the number of positive slots used in $s$ is $1$, then a bound of $2$ on the PoA of the instance follows trivially, so we assume there are at least $2$ positive slots used in the game. 

Let $t$ and $t'$ be used slots in $s$ such that $0<t<t'$. Let $x$ be the number of jobs on $t$ and $y$ the number of jobs on $t'$. Observe that $t$ must lie in the allowed interval of all jobs using $t'$ since $0<t<t'$ and these jobs use $t'$ in $s$ and $0$ in $s^*$. This means each of these jobs has the option to move to $t$ and pay the marginal contribution $c(x+1)-c(x)$. Also note that, by Lemma \ref{mechanism-lemma}, at least one of the jobs on $t'$ has to pay, in $s$, the average cost share $c(y)/y$. The above, combined with the equilibrium condition \eqref{mechanism-nash}, imply:
\begin{equation}\label{good-property}
c(x+1)-c(x) \ge \frac{c(y)}{y}.
\end{equation}
Note that:
\begin{align}\label{marginal}
c(x+1)-c(x)&=(x+1)^d-x^d\nonumber\\
&=x^{d-1}\frac{\left(1+\frac{1}{x}\right)^d-1^d}{\frac{1}{x}}\nonumber\\
&\ge x^{d-1}\lim_{x\to+\infty}\frac{\left(1+\frac{1}{x}\right)^d-1^d}{\frac{1}{x}}\nonumber\\
&=d x^{d-1}
\end{align}
Combining \eqref{good-property} with \eqref{marginal}, we get:
\[ d x^{d-1} \ge y^{d-1} \Rightarrow y\ge d^{\frac{1}{d-1}} x. \]
This suggests that every slot $t$ must have at least $d^{1/(d-1)}$ (which is always larger than $e$) times the number of jobs as every other slot in $[0,t)$. This in turn implies that as we move from the largest slot closer and closer to $0$, the number of jobs decreases by at least a factor $d^{1/(d-1)}$. We will write $\alpha=d^{1/(d-1)}$. Then, if $h$ is the number of jobs on the last occupied slot, we get that the cost on positive slots is at most:
\[ \sum_{j=0}^{+\infty}\left(\frac{h}{\alpha^j}\right)^d=h^d\sum_{j=0}^{+\infty}\left(\frac{1}{\alpha^d}\right)^j=\frac{h^d}{1-\frac{1}{\alpha^d}}=\frac{h^d}{1-d^{d/(1-d)}}. \]
Recall we have assumed that positive slots have at least as much cost as the non-positive ones, meaning the total cost of $s$ is at most:
\[ C(s) \le \frac{2h^d}{1-d^{d/(1-d)}}. \]
By the fact that we have $h$ jobs on the largest slot, there at least $h$ jobs in the game and we get:
\[ C(s^*) \ge h^d. \]
Taking the ratio gives:
\[ \frac{C(s)}{C(s^*)} \le \frac{2}{1-d^{d/(1-d)}}, \]
which is a constant for every given $d$.\qed
\end{proof}

\begin{remark}
For $c(x)=\sqrt{x}$ and for common slot instances with $n$ jobs, our coordination mechanism reduces the PoA from $4n^{1/4}/3$ to at most $4$.
\end{remark}

We note that for common slot instances there always exists a NE. In fact, again we prove that any optimal solution can be a NE with the correct payment vector.

\begin{theorem}
For $c(x)=x^d$ with $d<1$ and common slot instances, our coordination mechanism induces games such that, for every optimal assignment $s^*$, there exists a vector of payments $\xi$ such that $(s^*,\xi)$ is a NE. 
\end{theorem}
\begin{proof}
An optimal assignment will clearly place all jobs on the same slot. Charging each job the fair share $x^{d-1}$ results in a NE since $x^{d-1}<1$ with $1$ being the cost any job would have to pay to deviate to a different slot and open it.\qed
\end{proof}

Earlier in the section we proved that our coordination mechanism reduces the PoA to a constant for common slot instances. On the contrary, we show that for general instances, the PoA remains super-constant even for our mechanism.

\begin{theorem}\label{bad}
For a game with $n$ jobs and $c(x)=x^d$ with $d<1$, the worst case PoA of our mechanism is $\Omega(n^{d(1-d)})$.
\end{theorem}
\begin{proof}
Consider a large number of jobs $n$ such that $n^d$ is integer. Our instance has $n^d$ slots and $n^d$ jobs $1,2,\ldots,n^d$ such that job $j$ can only be scheduled on slot $j$. All other $n-n^d$ jobs can be scheduled on any slot. Let $s^*$ be the assignment where every one of the unrestricted jobs is scheduled in slot $1$. We get:
\[ C(s^*) = \left(n-n^d+1\right)^d+\left(n^d-1\right) = \Theta\left(n^d\right). \]
The first term of the sum comes from slot $1$ and the second term comes from slots $2,3,\ldots,n^d$ which hold one job each.

Now consider the following NE $s$. The unrestricted jobs are split equally among all slots, with each slot having $n^{1-d}-1$ of them. The payments declared by the unrestricted jobs are $0$ and the full cost of each slot is paid for by the corresponding restricted job. This outcome is a NE since the unrestricted jobs get to freeload while the restricted jobs can't move to a different slot and have to pay enough to keep their slot open and avoid the large cost of remaining unscheduled. We get:
\[ C(s) = n^d \left(n^{1-d}\right)^d = \Theta\left(n^{2d-d^2}\right). \]
Taking the ratio $C(s)/C(s^*)$ completes the proof.\qed
\end{proof}
\section{Conclusion and Open Problems}\label{conclusion}

In this work we tackled the problem of load-dependent server activation costs in real-time scheduling systems, a question that was posed as an open problem in \cite{T18}. We precisely characterized the price of anarchy for monomial cost functions. Furthermore, we came up with a novel coordination mechanism that achieved a spectacular improvement of the price of anarchy in the original model of unit activation costs and in a restricted subclass of instances in our extended model.

There are several follow-up questions that emerge from our work. These include extending our results to wider classes of cost functions and settling the PoA gap between the constant bounds for $d>1$ (e.g., between the $2.00568$ lower bound we prove in this work and the standard $2.5$ upper bound for $d=2$). Most importantly, it is interesting to see whether coordination mechanisms of the type that we define here can yield similar price of anarchy improvements in other important problems in algorithmic game theory or in the setting where the current incarnation of the mechanism fails in our model: monomial costs with $d<1$. Our mechanism relies on the simple idea of offloading the cost sharing aspect of the problem to the jobs themselves. This induces a setting where any deviation is charged the marginal contribution of the job on the new slot. Variants of the mechanism that impose, e.g., upper bounds on the cost share of a job/player might prove useful in these endeavors (for instance, such a modification would break the bad example in Lemma \ref{bad}).

\bibliographystyle{abbrv}
\bibliography{cost-sharing-bib}

\end{document}